\documentclass[aps,prl,twocolumn,nofootinbib,superscriptaddress,10pt]{revtex4-2}

\pdfoutput=1

\usepackage{mathtools,amsmath,amssymb,amsfonts,amsthm, textcomp}
\usepackage{color}
\usepackage{physics}
\usepackage{graphicx}
\usepackage{wrapfig}
\usepackage{flafter}
\usepackage{hyperref}
\usepackage{inputenc}
\usepackage{thmtools}
\usepackage{thm-restate}

\newtheorem{lemma}{Lemma}

\begin{document}

\title{Causal UV completions of relativistic hydrodynamics}

\author{Robbe Brants}
\email{robbe.brants@ugent.be}
\affiliation{Department of Physics and Astronomy, Ghent University, 9000 Ghent, Belgium}

\begin{abstract}
    Relativistic hydrodynamics successfully provides an effective field theory description for the low energy regime of many out-of-equilibrium systems. On the other hand, in this paper we prove that any stand-alone hydrodynamic EFT is inherently acausal and therefore requires the addition of transient UV modes in order to restore causality. This is made possible by the exponential decay of dissipative hydrodynamics in a majority of the lightcone, allowing the possibility of a causal description that still reduces to the hydrodynamic one at late timescales. We then investigate the emergence and possible restrictions of the non-hydrodynamic modes in these causal UV completions.
\end{abstract}

\maketitle

\section{Introduction}
\vspace{-5pt}
Relativistic hydrodynamics plays a crucial role in the phenomenological description of various out-of-equilibrium systems \cite{Florkowski:2017olj, Andersson:2020phh, Hammond:2021vtv}. Its success is regarded as a result of the separation of scales between the modes governing the dynamics of the system. A hydrodynamic mode corresponds to a conserved quantity in the system, and as a result its dispersion relation $\omega(k)$ in complex frequency space vanishes in the limit of infinite wavelength $k \to 0$. In contrast, the remaining degrees of freedom correspond to transient or non-hydrodynamic modes, which decay exponentially in time at all wavelengths. Relativistic hydrodynamics can therefore be seen as an EFT description, applicable in the IR limit corresponding to large wavelengths and late times, when the system is close to equilibrium.

On the other hand, this success may still come as a surprise, since this EFT description is expected to be acausal without corrections in the form of non-hydrodynamic modes, a clear flaw when describing a relativistic fluid \cite{HISCOCK1983466,Hiscock:1985zz,Gavassino:2021owo,Gavassino:2023mad}. This statement will be made explicit in the paper, but already a straightforward example is Fick's diffusion equation, with hydrodynamic mode $\omega(k) = -iDk^2$. After perturbing the system at a single point in space, the retarded correlation function governed by this mode spreads out instantaneously throughout the entirety of real space for all $t>0$:
\begin{equation} \label{eq:Gdiff}
    G(t,x) = \frac{1}{2\pi}\int e^{ikx-i\omega(k)t}dk = \frac{1}{\sqrt{4\pi D t}}\exp\left(-\dfrac{x^2}{4Dt}\right)\,.
\end{equation}
A possible solution to this problem is moving away from the IR limit, and reintroducing the UV modes that were integrated out. An example is given by the Telegrapher's equation, where a transient mode with relaxation time $\tau_\pi$ is added. This type of causal completion is used in many hydrodynamic descriptions such as MIS theory \cite{muller1967,ISRAEL1976310,Israel:1979wp} and the BDNK framework \cite{Bemfica:2017wps,Kovtun:2019hdm,Bemfica:2019knx}. However, this approach also modifies the dispersion relation of the hydrodynamic mode in the IR regime, in the form of higher gradient corrections. \\
\\
We instead want to focus on \emph{causal UV completions}, which are causal correlation functions containing additional transient UV modes, that have any particular hydrodynamic dispersion relation as an \emph{exact} mode below some energy cutoff, where it furthermore remains the slowest decaying mode.

This approach can be seen as a bootstrap method for constructing the UV completion of an EFT. Here, the IR description is fixed by the hydrodynamic dispersion relation, while we remain fully agnostic about the microscopics of the system. We then use causality as an overarching principle to put restrictions on the regime of applicability of hydrodynamics, together with the compatible UV completions. \\
\\
This goal is achieved by proving three mathematical theorems, together with a proper study of causal mode structures. First, we investigate the requirements for correlation functions and their modes to be causal, and proof stand-alone dissipative hydrodynamic theories are guaranteed to violate causality, justifying the need for a causal completion.
We then proceed to look at the effects of dissipative hydrodynamics on real-space evolution, where we find a decay rate depending on the radial direction (velocity) within the lightcone. This is used to guarantee the existence of causal UV completions for hydrodynamic theories.
Finally, we look at the emergence of non-hydrodynamic modes in these UV completions, in particular for the example of pure diffusion. \\
\\
In the paper, correlation functions and Fourier transforms will be considered as if one-dimensional. Generalizing to higher dimensions is reasonably straightforward, as it adds polynomial factors while we will focus entirely on exponential decay rates.

\vspace{-5pt}
\section{Causal correlation functions}
\vspace{-5pt}
The main object of interest in this paper will be the retarded correlation function $G(t,x)$, which describes the reaction of a fluid to an initial perturbation at $x = 0$, together with its singularities in complex frequency space, known as the modes of the system. If $G(t,x)$ is a causal correlation function, it must vanish outside of the forward light cone, i.e. when $|x|>t$. In QFT, $G(t,x)$ is generally assumed to be a tempered distribution. In simple terms, this means its growth is bounded by some polynomial at large $x$ and $t$, and so this condition translates to stability of the theory. It furthermore guarantees the existence of a Fourier transform
\begin{equation}
    G(\omega,k) = \int G(t,x) e^{i\omega t-ikx} dtdx \,,
\end{equation}
and because of the causality condition, we find that the integrand is exponentially suppressed when $\Im \omega > |\Im k|$. As a result, the integral and all its derivatives will converge absolutely in this region of Fourier space. This leads to the well studied causality condition that all singularities of the retarded correlation function must obey \cite{Heller:2022ejw,Heller:2023jtd,Gavassino:2023myj,Hui:2025aja,Roy:2026tbr}:
\begin{equation} \label{eq:causal cond}
    \Im\omega(k) \le |\Im k| \,.
\end{equation}

This condition is sufficient to bring us to the first important point of this paper, which is that dissipative hydrodynamics as a stand-alone EFT can \emph{not} be causal, meaning that if only the hydrodynamic mode contributes to the correlation function, it will have support outside of the lightcone. In this statement, a dissipative hydrodynamic dispersion relation corresponds to a function $\omega(k):\mathbb C \to \mathbb C$, for which $\omega(0) = 0$ and $\Im \omega(k) < 0$ for all $k \in\mathbb R\backslash \{0\}$. Explicitly, we have the following result:
\begin{restatable}[Acausality of hydrodynamics]{thm}{acausal hydro} \label{th:acausal hydro}
    Let $\omega(k)$ be a dissipative hydrodynamic dispersion relation. Then, the correlation function
    \begin{equation}
        G(t,x) = \frac{1}{2\pi}\theta(t)\int e^{i kx-i\omega(k)t}dk
    \end{equation}
    has support outside of the lightcone $|x| \le t$.
\end{restatable}
\begin{proof}
    We first observe that, since $\Im\omega(k) \le 0$ for all $k\in \mathbb R$, $\theta(t)e^{-i\omega(k)t}$ is bounded for all $t$. It is therefore a tempered distribution, and so is the correlation function $G(t,x)$ at a fixed $t$. For any Schwartz function $\varphi\in\mathcal S(\mathbb R^2)$, we can make use of the boundedness of the integrand to find
    \begin{equation}
        \langle G,\varphi\rangle = \frac{1}{2\pi} \int_0^\infty \int_{-\infty}^\infty \overline{\hat\varphi(t,k)} e^{-i\omega(k)t} dt dk \,,
    \end{equation}
    and therefore
    \begin{equation}
        |\langle G, \varphi \rangle| \le \frac{1}{2\pi} \int_0^\infty \int_{-\infty}^\infty |\hat\varphi(t,k)| dt dk < \infty \,,
    \end{equation}
    meaning $G(t,x)$ is also tempered in $\mathbb R^2$.
    
    We can thus directly calculate its Fourier transform
    \begin{equation}
        G(\omega,k)=\frac{i}{\omega-\omega(k)} \,.
    \end{equation}
    We will now assume the correlation function to vanish outside of the lightcone, and look separately at the cases where $\omega(k)$ is either entire, or when it contains singularities. If $\omega(k)$ is an entire function, we can use Theorem 1 from \cite{Heller:2022ejw} which states that, in order to satisfy the causality condition \eqref{eq:causal cond}, $\omega(k)$ must be a polynomial of at most degree one. Since $\omega(k)$ is dissipative, this is not a possibility. The dispersion relation can therefore not be analytic in the entire complex $k$-space. We invert the relation
    \begin{equation}
        \omega(k) = \omega - \frac{i}{G(\omega,k)} \quad \forall(\omega,k)\in\mathbb C^2\,,
    \end{equation}
    and see that $\omega(k)$ will be non-analytic at $k=k_*$ if either $G(\omega,k)$ is also non-analytic at $k_*$, or $G(\omega,k_*)=0$. In the second case, $\omega(k)$ has a pole, which is ruled out from Theorem 2 in \cite{Heller:2022ejw} under the causality condition. In the first case, $G(\omega,k)$ has a singular point in complex $k$-space, whose location is crucially independent of $\omega$. Now remember that the causality bound counts for any pair of values $(\omega,k)$ where $G(\omega,k)$ is singular. In this scenario, we can always choose a value of $\omega$ with $\Im \omega > |\Im k_*|$, such that $G(\omega,k_*)$ would need to be analytic if $G(t,x)$ is causal. With this final contradiction, we conclude that $G(t,x)$ must have support outside of the lightcone.
\end{proof}

This theorem lays the foundation for the rest of the paper, since it enforces the need of a UV completion for any hydrodynamic theory other than a perfect fluid. 

The arguments in this proof are not novel \cite{Heller:2022ejw, Gavassino:2023mad}, but it is important to stress they follow entirely from the causal restriction \eqref{eq:causal cond} on \emph{all} singularities of the correlation function, not just the modes in frequency space. 

Even if a hydrodynamic dispersion relation satisfies the causality condition through singularities in $k$-space, it still cannot exist in a causal correlation function without the addition of other modes that remove these singularities. A perfect example is the hydrodynamic mode in Israel-Stewart theory:
\begin{equation}
    \omega(k) = -\frac{i}{2\tau_\pi} \left(1-\sqrt{1- 4D\tau_\pi k^2}\right) \,.
\end{equation}
It satisfies the causality condition, but due to the branch cut singularity at $k=\pm1/\sqrt{4D\tau_\pi}$, the single mode correlation function still has support outside the lightcone. This support is exactly canceled when adding the transient mode to the correlation function.

In general, to ensure no additional singularities occur other than the modes in frequency space, it is sufficient to have local equations of motion \cite{Gavassino:2023mad}, since this turns the denominator of the correlation function into a polynomial in $\omega$ and $k$.

\vspace{-5pt}
\section{Hydrodynamic decay}
\vspace{-5pt}
As the imaginary part of $\omega(k)$ vanishes for large wavelengths, one may expect the correlation function to decay subexponentially. But for a diffusive mode for example, this is only true when looking at the response function for a fixed point $x$ in space. When tracking the response along a radial line $x = vt$ with velocity $v$, we still observe an exponential decay for all $v$ other than the sound speed defined as $c_s := \omega'(0)$. For the diffusive mode, where $c_s=0$, we find from \eqref{eq:Gdiff} that
\begin{equation}
    G(t,v)\equiv G(t,x=vt) \propto \exp\left(- \frac{v^2}{4 D} t\right) = e^{-\Gamma(v)t}
\end{equation}
decays at a rate $\Gamma(v) >0$ for all $v\ne 0$. For more complicated dispersion relations, it is typically not possible to find a closed form expression for the correlation function, but if $\omega(k)$ is analytic in a region $\mathcal R \subset \mathbb C$ around $k=0$, we can still compute the late time decay rate using a saddle-point (also known as steepest descent) method \cite{fedoryuk1989asymptotic}. 

This method works as follows. Let $S(k,v) := ikv-i\omega(k)$ and consider the hydrodynamic correlation function
\begin{equation}
    G(t,v) = \frac{1}{2\pi}\theta(t) \int_{-\infty}^\infty e^{S(k,v)t}dk \,.
\end{equation}
We notice that the integrand is analytic in $\mathcal R$, so we can equivalently compute this integral along a different contour, as long as the deformation remains within $\mathcal R$. Our goal is to find the contour $\gamma$ along which the real part of $S(k,v)$ remains as small/negative as possible:
\begin{equation}
    M(v) = \inf_\gamma\sup_{k\in\gamma} \Re S(v,k) \,.
\end{equation}
Let us say $\Re S(v,k)$ reaches this value $M(v)$ at a set number of critical points $\{k_n\}$. Due to the analyticity of $S$ in ths region, these critical points are necessarily either the end points of the contour, or saddle points. We will for now focus on the latter. The contribution of all other parts of the contour will be exponentially suppressed compared to that of the regions near these saddle points. If we further assume the saddles to be simple (i.e. $\partial_k^2S(k,v)|_{k\in \{k_n\}}\ne0$), we find the late time approximation:
\begin{multline}
    G(t,v) 
    \sim \frac{1}{\sqrt{2\pi t}}e^{M(v)t}\sum_n [-\partial_k^2S(k_n,v)]^{-1/2} \\ \equiv \frac{C(v)}{\sqrt{t}} e^{-\Gamma(v)t} \,.
\end{multline}
At $v=c_s=\omega'(0)$, the point $k=0$ will be a saddle point of $S(k,c_s)$, and since $S(0,v)=0$, we also find $\Gamma(c_s)=0$. For any other value of $v$, $\partial_kS(0,v)\ne0$, and so we can (formally) deform the contour towards one where $S(k,v) < 0$ everywhere on the contour, leading to an overall positive decay rate $\Gamma(v\ne c_s)>0$. Until know, we have neglected the end points of the contour. In the case where $\lim_{k\to\pm\infty}\omega(k) = 0$, these points will lead to a non-decaying correlation function. In order to avoid this, we need the stronger requirement that $\omega(k)$ maintains a finite imaginary gap away from $k=0$:
\begin{equation} \label{eq:gapped hydro}
    (\forall \delta>0)(\exists \epsilon>0)(\forall k\in\mathbb{R}\backslash[-\delta,\delta])(\Im\omega(k) \le -\epsilon) \,.
\end{equation}
This previous intuition can be formulated in the form of the following theorem.
\begin{restatable}[Decay rate]{thm}{Decay} \label{th:decay}
    Let $\omega(k) : \mathbb C \to \mathbb C$ and $g(k) : \mathbb C \to \mathbb C$ be analytic for $|k|<R$. $|g(k)|$ is furthermore bounded by a polynomial on $\mathbb R$. Let $c_s=\omega'(0)$ and
    \begin{equation}
        G(t,v) = \frac{1}{2\pi}\theta(t)\int_{-\infty}^\infty g(k) e^{ikvt-i\omega(k)t}dk \,.
    \end{equation}
    If $\omega(k)$ satisfies the condition \eqref{eq:gapped hydro}, then there exists a tempered distribution $F(t,x) \in \mathcal S'(\mathbb R^2)$ and smooth function $\Gamma(v)$ such that:
    \begin{itemize}
        \item $G(t,v) = F(t,vt)e^{-\Gamma(v)t}$ for all $t>0$ and $v\in \mathbb R$,
        \item $\Gamma(v) > 0$ for all $v\ne c_s$,
        \item $\Gamma(v)$ increases monotonically in $|v-c_s|$.
    \end{itemize}
\end{restatable}
\noindent
The proof is given in the Appendix.

It is worth noting that the monotonic function $\Gamma(v)$ is not necessarily the precise exponential decay of $G(t,v)$, since the latter is generally not smooth in $v$. Furthermore, when $\Im\omega(k)$ is not monotonic in $|k|$, it is possible that $G(t,v)$ admits exponential decay that becomes \emph{slower} when $|v-c_s|$ increases. In this case, $\Gamma(v)$ is a lower bound on the exponential decay, that \emph{does} remain monotonic. The existence of this lower bound is a result of the finite, negative imaginary part of $\omega(k)$ for $k\ne 0$ even in the limit $k\to\pm\infty$. 

When $G(t,v)$ is determined by the sum over a (finite) number of modes $\omega_n(k)$, the generalization is straightforward. Each contribution is described by its own distribution $F_n(t,x)$ and exponential decay $\Gamma_n(v)$, and so the decay of the total correlation function is given by the slowest decaying mode at each value of $v$:
\begin{multline}
    G(t,v) = \sum_n F_n(t,x=vt)e^{-\Gamma_n(v)t} \\
    = \left\{\sum_n F_n(t,x)e^{-[\Gamma_n(v)-\Gamma(v)]t}\right\}e^{-\Gamma(v)t} \\
    \equiv F(t,x) e^{-\Gamma(v)t} \,,
\end{multline}
where $\Gamma(v)$ is a smooth function with the same properties as $\Gamma_n(v)$, that lower bounds $\min_n \{\Gamma_n(v)\}$ at each $v\in\mathbb R$. As the sum over products of tempered distributions with smooth, bounded functions, $F(t,x)$ is also a tempered distribution.
\\
\\
From Theorem \ref{th:decay}, we learn that at late times the hydrodynamic contribution is concentrated in a small region around the radial direction $v=c_s$, since the correlation function decays exponentially in time everywhere else. Furthermore, we observe that the late time dynamics of the correlation function are not necessarily described by hydrodynamics at all points in the lightcone, even if the UV sector is transient. This happens when $\Gamma(v)$ exceeds the decay rate of the UV sector, which can be derived in an analogous manner.

\section{Causal UV completions}
\vspace{-5pt}
The result from the previous section is what allows us to consider causal completions of hydrodynamics where all the modes are added in the UV regime, without needing to alter the hydrodynamic dispersion relation. Instead, we simply remove the mode outside of some region $\mathcal K $ in $k$-space, that contains the origin $k=0$. \\
\\
Intuitively, this works as follows. Consider a dispersion relation $\omega(k)$ with sound speed $|c_s| < 1$, that is dissipative at small values of $k$. If the mode would become unstable at some real value of $k$, we can simply cut the regime short and set $\omega(k) = -i\infty$ outside of it. We now apply Theorem \ref{th:decay} and find that the contribution of the hydro mode decays exponentially in time outside of the light cone, since $|v|\ge1>|c_s|$. We then remove these exponentially decaying tails outside of the lightcone and redistribute their contribution inside of it (for conservation reasons). Since these corrections decay exponentially in time, they will only add UV modes with finite $\Im \omega < 0$. This leaves the hydro mode entirely intact for sufficiently small values of $|k|$. 

This leads us to:

\begin{restatable}[Causal UV completion]{thm}{compl} \label{th:compl}
    Let $\omega(k):\mathbb C \to \mathbb C$ be analytic for $|k|<R$, and have $\omega(0)=0$, while $\Im \omega(k) < 0$ in a real, punctured neighborhood of $k=0$. If $|c_s|=|\omega'(0)| < 1$, there exists an open interval $0\in \mathcal K \subseteq \mathbb R$ and a correlation function $G(\omega,k)$, for which
    \begin{itemize}
        \item its inverse Fourier transform $G(t,x)$ vanishes when $|x|>t$.
        \item For all $k\in\mathcal K$, the slowest decaying mode (largest imaginary part) is a single pole at $\omega = \omega(k)$.
    \end{itemize}
\end{restatable}
\begin{proof}
    Firstly, we choose a value of $\tilde R < R$, such that $\Im \omega(k) < 0$ for all $k \in [-\tilde R,R]\backslash{0}$, and define
    \begin{equation}
        G_0(\omega,k) := \theta(\tilde R^2-k^2) \frac{i}{\omega-\omega(k)} \,.
    \end{equation}
    This distribution is compact in $k$, and bounded for large $\omega$, so it is tempered. We can therefore compute its inverse Fourier transform
    \begin{equation}
        G_0(t,x) = \frac{1}{2\pi}\theta(t)\int_{-\tilde R}^{\tilde R} e^{ik x-i\omega(k) t}dk \,.
    \end{equation}
    This distribution is not only tempered, but also analytic for all $x\in\mathbb R$ and $t\ge0$ due to the compact support in $k$. This allows us to define the 'causal completion'
    \begin{equation}
        G(t,x) := G_0(t,x) \theta(t^2-x^2) \,,
    \end{equation}
    which, as the product of a bounded function with a smooth\footnote{For simplicity, we ignore the $\theta(t)$ factor that turns the function discontinuous, since it is already incorporated in $\theta(t^2-x^2)$.}, tempered function is a tempered distribution. The same is true for
    \begin{equation}
        \Delta(t,x) := G_0(t,x) - G(t,x) \,.
    \end{equation}
    From Theorem \ref{th:decay}, we find that $G_0(t,x)$ will admit an exponential decay with rate $\Gamma(v)$, which is monotonically increasing with $|v-c_s|$ and nonzero for $v\ne c_s$. Since $|c_s|<1$, we find that $\Gamma(v) \ge \min\{\Gamma(\pm1)\} =: \Gamma_m > 0$ for all $|v| \ge 1$. Therefore, we can write
    \begin{equation}
        \Delta(t,x) = \Delta_\Gamma(t,x) e^{-\Gamma_m t} \,,
    \end{equation}
    where $\Delta_\Gamma(t,x)$ is tempered. Its Fourier transform is
    \begin{multline}
        \Delta(\omega,k) = \int_0^\infty \int_{-\infty}^\infty \Delta(t,x) e^{i\omega t-i k x}dt dx \\
        = \int_0^\infty\tilde\Delta_\Gamma(t,k) e^{(i\omega-\Gamma_m)t}dt \,.
    \end{multline}
    As the Laplace transform of a tempered distribution, we find $\Delta(\omega,k)$ will be analytic when $\Im\omega > -\Gamma_m$. Since $\omega(0)=0$, there will be an interval $\mathcal K \subseteq (-\tilde R,\tilde R)$ where this inequality is satisfied by the mode $\omega(k)$. To conclude, within this interval $\mathcal K$,
    \begin{equation} \label{eq:Gcomp def}
        G(\omega,k) = G_0(\omega,k) - \Delta(\omega,k)
    \end{equation}
    will have no singularities with $\Im\omega > -\Gamma_m$, other than the single pole at $\omega = \omega(k)$ arising from $G_0(\omega,k)$.
\end{proof}
The power of this Theorem lies in the implication that \emph{any} hydrodynamic dispersion relation can exist within a causal UV completion, given the sound speed is subluminal and that the hydrodynamic gradient expansion has finite radius of convergence. The first requirement is fully physical and presumably impossible to get rid of, since a dissipative mode with sound speed larger or equal to one will violate the causality condition \eqref{eq:causal cond} in any complex neighborhood of $k=0$. The second requirement is mathematical in nature, and likely not strictly necessary. Having it as a necessity of the proof in its current form unfortunately does exclude incorporating effects from stochastic fluctuations \cite{Kovtun:2012rj, Chen-Lin:2018kfl}. \\
\\
A first remark on the proof is the restriction of $k$ to a finite interval of the real line. This is important when a mode would otherwise become unstable, as is for example the case for the particle diffusion mode in RTA kinetic theory \cite{Romatschke:2015gic}:
\begin{equation}
    \omega(k) = -\frac{i}{\tau_R}\left(1-\frac{k\tau_R}{\tan(k\tau_R)}\right) \,,
\end{equation}
where removing the mode corresponds to absorption by a branch cut \cite{Brants:2024wrx}.

In many other cases however, using a finite interval is not actually required, but using the full real line has the consequence that $G_0(t,x)$ is now no longer analytic, and multiplying with a discontinuous function $\theta(t^2-x^2)$ may pose issues. It is however possible to show that, when $\lim_{k\to\pm\infty}\Im\omega(k) = -\infty$ faster then $\Im\omega(k)\sim -|k|$, and $\omega(k)$ is analytic in a strip around the real line (as such, any stable polynomial suffices), $G(t,x)$ can still be defined rigorously. This approach has the advantage that we can now extent the real interval $\mathcal K$ to a complex one within the analytic strip. A proof of this statement is particularly convoluted and goes beyond the scope of the paper, since it arguably does not hold more physical relevance. \\
\\
Finally, we want to point out that we removed a part of the real-space correlation function, hence possibly violating the conservation that originally lead to the emergence of hydrodynamics. It is however trivial to add this contribution back onto the edges of the lightcone, resulting in UV modes on the lines $\omega \in  \pm k -i [\Gamma_m,\infty)$; a perfect example is given by \cite{Gavassino:2026zsz}. Alternatively, the contribution can be spread out over the full light cone, which happens in many RTA examples \cite{Bajec:2024jez}. The details of this redistribution are entirely encoded in the UV modes of the system, and as such are not accessible from the hydrodynamic dispersion relation alone. 

\vspace{-1pt}
\section{Non-hydrodynamic modes}
\vspace{-5pt}
In this section, we discuss the modes that emerge from a causal completion, and how we can bootstrap the non-hydrodynamic sector based on the hydrodynamic dispersion relation. This question can have both simple and extremely complicated answers. Let us assume we have a way of computing the \emph{exact} decay rate $\Gamma_\text{H}(v)$ of the hydrodynamic sector, not just the lower bound used in previous sections.

On one hand, we could argue that it is possible to add just about any non-hydrodynamic mode to the mode structure as long as they do not violate causality. In this case, all we can say about the non-hydrodynamic sector is that it must be sufficiently long-lived, i.e. we require modes with
\begin{equation}
    -\frac{1}{\tau_R}\equiv\sup_{k\in\mathbb R} \left\{\Im \omega_\text{NH}(k)\right\} \ge -\lim_{v\to 1^+}\Gamma_\text{H}(v)
\end{equation}
in order to be able to cancel the contribution of the hydrodynamic mode just outside of the lightcone. This sets a direct lower bound on the thermalization timescale of the non-hydrodynamic sector. Of course, this is only a necessary condition and the modes emerging from $\Delta(\omega,k)$ in \eqref{eq:Gcomp def} will be much more explicit. 

But the point to be made here, is that without any other assumptions, the contribution from these modes, together with the hydrodynamic mode, can in theory be completely swamped by arbitrarily long-lived modes that may be added on top of the ones required for causality. 

In this scenario however, the EFT approach to hydrodynamics does not apply in the first place. The success of hydrodynamics in describing late time dynamics can thus be used as a motivation to instead focus on the minimal UV corrections needed to restore causality. In particular, we want a large part of the lightcone to behave ``hydrodynamically'' at late times, meaning $G(t,v) \approx G_0(t,v)$ in the limit $t\to \infty$ and $|v|<1$. \\
\\
Let us return to the example of pure diffusion, since it has an analytic solution. In this case, we had $\Gamma_\text{H}(v) = v^2/(4D)$, meaning the maximal gap of the non-hydrodynamic sector is $\Im\omega = -1/(4D)$ at real $k$. The thermalization timescale towards the edge of the lightcone is therefore $\tau_R \ge 4D$ for pure diffusion, a more explicit version of the bound found in \cite{Hartman:2017hhp}. 

For this dispersion relation we can explicitly compute $\Delta(\omega,k)$, and find it indeed has two branch cuts at
\begin{equation}
    \omega \in \pm k-i\left[\frac{1}{4D},\infty\right) \,.
\end{equation}
In order to keep the signal inside of the lightcone, we choose to add the removed contributions back onto the edge in the form of two peaks $\frac{1}{2}\Delta(t)\delta(x^2-t^2)$. Since $\Delta(t)$ also decays at a rate $\ge \Gamma_\text{H}(1)$, this will lead to the same type of branch cuts. This choice of re-adding the modes through ballistic transport is physically motivated for massless particles in one spatial dimension. In higher dimensions or with massive transport we would also expect a continuum of modes with $-k < \Re \omega < k$ \cite{Brants:2024wrx,Gavassino:2026zsz}. Altogether, we obtain the following causal UV completion of diffusion:
\begin{widetext}
\begin{equation}
    G(\omega,k) = \frac{i}{\omega+i Dk^2} \left(\frac{1-2iDk}{2\sqrt{1-4iD(k+\omega)}} \right) - \frac{i}{\omega^2-k^2} \left(\frac{\omega-k}{2\sqrt{1-4iD(k+\omega)}}-\frac{\omega}{2}\right)+(k\leftrightarrow -k) \,.
\end{equation}
\end{widetext}
The first term contains both the removal of the tails, as well as the hydrodynamic mode which exists for all $|\Im k| < 1/(2D)$, see also \cite{Gavassino:2026tvy}. The second term is added to ensure conservation. Note that this mode structure is fully identical to the one found in \cite{Gavassino:2026zsz} for the one-dimensional case. This should not be surprising at all, because the entire real-time correlation function is uniquely determined by the dispersion relation, causality and conservation! For all $|x|<t$, the correlation function is equal to the hydrodynamic $G_0(t,x)$, while the peaks at $x=\pm t$ are fixed by conservation (and $x\leftrightarrow -x$ reflection symmetry). 

This is a large contrast compared to the previous statement and EFT principle that the UV regime exists independent of the IR regime. In general, we would expect a result that falls in between the two extremes: part of the UV modes are fixed by the causal completion, while the other part is distributed over the complex plane, adding non-trivial transient behavior. This can be done using an EFT construction with additional guiding principles, based on the properties of the system \cite{Liu:2018kfw,An:2025ils}.

\vspace{-5pt}
\section{Conclusions and outlook}
\vspace{-5pt}
So far, we have laid the basis of how to bootstrap the UV sector from of a hydrodynamic EFT using causality as an overarching principle, on the level of retarded correlation functions. The fact this additional sector is necessary in the first place, is proven by Theorem \ref{th:acausal hydro}, which shows that dissipative hydrodynamics as a stand-alone theory cannot be causal. 
We furthermore demonstrated in Theorem \ref{th:compl} that such a causal UV completion will always exist for hydrodynamic dispersion relations with subluminal sound speed and finite radius of convergence of the hydrodynamic expansion. The combination of these two theorems provides a powerful justification for studying this field.

We also proved that dissipative hydrodynamics will lead to exponential decay in the majority of the lightcone, and provided a practical method of computing the decay rate. This is an important feature for predicting properties of the non-hydrodynamic sector, such as a minimal thermalization timescale similar to \cite{Delacretaz:2023pxm, Qi:2026vht}. This is demonstrated explicitly for diffusion, but can be easily extended to more complicated dispersion relations. \\
\\
Our results have important implications when considering causal restrictions on hydrodynamics, as discussed in terms of the \emph{hydrohedron} \cite{Heller:2023jtd}. Since practically any dispersion relation can be part of a causal UV completion, it is clear that one cannot directly restrict hydrodynamics. Instead we restrict the regime of applicability of hydrodynamics, given by $\mathcal K$ in Theorem \ref{th:compl}, which is the more general interpretation of $|k|<R$ in aforementioned paper. This implication agrees with the results and discussion in \cite{Gavassino:2026fil}, but an important difference is that we additionally require the hydrodynamic mode to be decaying slowest, in order for hydrodynamics to be applicable at late times. \\
\\
Finally, we note that the exponential decay rate $\Gamma(v)$ also provides stronger causal bounds of the form
\begin{equation}
    \Im \omega \le \sup_{|v|<1} [v \Im k -\Gamma(v)] \,,
\end{equation}
allowing for a possible improvement of the hydrohedron bounds, which is a promising direction for future work. \\

\noindent
\textbf{\emph{Acknowledgements} -}
First and foremost, I want to thank Michal P. Heller for the supervision of this work and his invaluable feedback and discussions in the writing of this paper. I furthermore want to thank Lorenzo Gavassino and Alexandre Serantes for the useful discussions leading up to this work, as well as Matisse De Lescluze and Xin An for their feedback.

This project has received funding from the European Research Council (ERC) under the European Union’s Horizon 2020 research and innovation programme (grant number: 101089093 / project acronym: High-TheQ). Views and opinions expressed are however those of the authors only and do not necessarily reflect those of the European Union or the European Research Council. Neither the European Union nor the granting authority can be held responsible for them.

\bibliographystyle{bibstyl}
\bibliography{main}

\section{Hydrodynamic decay - proofs} \label{sec:decay proofs}

We start off with a specialized version of the deformation lemma \cite{chabrowski1996introduction} for analytic functions, where we restrict the contour and its deformation such that their projection on the real axis $\gamma \mapsto \Re(\gamma)$ is bijective. Although this restriction is often not necessary, it will significantly simplify future theorems.

\begin{lemma}[Deformation Lemma] \label{lm:deformation}
    Let $y(x)\in C^1(\mathbb R,\mathbb R)$ define a contour $\gamma = \{x+iy(x)| x\in \mathbb R\} \subset \mathbb C$, and let $f:\mathbb C\to\mathbb C$ be holomorphic in $\mathcal R\subseteq \mathbb C$. If there exists a point $z_0 \in \gamma \cap \mathcal R$ for which $\Im f'(z_0) > 0$, then there exists a contour deformation $\tilde y(x)\in C^1(\mathbb R,\mathbb R)$, $\tilde\gamma = \{x+i\tilde y(x)| x\in \mathbb R\}$, and a neighborhood $V \subseteq \mathbb R$ of $\Re z_0$ for which
    \begin{enumerate}
        \item $y(x) \le \tilde y(x) \quad \forall x \in \mathbb R$
        \item $\{z\in\mathbb C: y(\Re(z)) \le \Im z \le \tilde y(\Re(z)) \} \backslash (\gamma\cap\tilde\gamma) \subset \mathcal R$
        \item $\sup_{z\in\tilde\gamma} \Re f(z) \le \sup_{z\in\gamma} \Re f(z)$
        \item $\sup_{x\in V} \Re f(x+i\tilde y(x)) < \sup_{x\in V} \Re f(x+iy(x))$ .
    \end{enumerate}
\end{lemma}
\begin{proof}
    Firstly, because $f$ is holomorphic in $z_0$ and $\Im f'(z_0) > 0$, there exists a neighborhood $U_0 \subseteq \mathcal R$ of $z_0$ in which $\Im f'(z) > 0$. If necessary, we can further shrink $U_0$ to a compact, convex neighborhood $U \subseteq U_0$ of $z_0$ such that
    \begin{equation}
        z\in U \Rightarrow \Re z+i[y(\Re(z)),\Im(z)] \subset U \cap \mathcal R  \,.
    \end{equation}
    What this condition ensures, is that within $U_1$ we can deform the contour in the positive imaginary direction, without leaving $\mathcal R$, and that at at each point in this region $\Re f(z)$ strictly decreases with increasing $\Im z$. The latter is a result of of $f(z)$ being holomorphic in $U$:
    \begin{equation}
        \lim_{h\to 0} \frac{\Re f(z+ih)-\Re f(z)}{h} = -\Im f'(z) < 0 \,.
    \end{equation}
    Let $x_1 := \sup\{x\in \Re(\gamma \cap\partial U): x < \Re z_0\}$ and $x_2 := \inf\{x\in \Re(\gamma \cap\partial U): x > \Re z_0\}$ with $\partial U$ the boundary of $U$ in complex space. The upper boundary $\partial U^+(x) \equiv \sup \{y \in \mathbb R | x+iy \in U\}$ then defines a function on $[x_1,x_2]$, with $\partial U^+(x) > y(x)$ for all $x\in (x_1,x_2)$. $\partial U^+(x_1)$ and $\partial U^+(x_2)$ do not necessarily match $y(x_1)$ and $y(x_2)$, but we can straightforwardly introduce a continuous function $\tilde y(x)$ on $(x_1,x_2)$ such that $y(x) < \tilde y(x) \le \partial U^+(x)$, while $\lim_{x\to x_1}\tilde y(x) = y(x_1)$ and $\lim_{x\to x_2}\tilde y(x) = y(x_2)$. An example is
    \begin{equation}
        \tilde y(x) = y(x) + \left[\partial U^+(x)-y(x)\right]\exp\left[\frac{1}{(x-x_1)(x-x_2)}\right] .
    \end{equation}
    Outside of $(x_1,x_2)$, we set $\tilde y(x) = y(x)$.
    
    The first condition now follows directly from the construction of $\tilde y$, while the second and third are a result of the properties of $U$. For each $x\in(x_1,x_2)$, $\Re f(x+i\tilde y(x)) < \Re f(x+i y(x))$ has strictly decreased. We can therefore simply choose $V = (\tilde x_1,\tilde x_2)$ with $x_1<\tilde x_1<\tilde x_2<x_2$, which ensures the final condition.
\end{proof}

The proof is analogous if $\Im f'(z_0) < 0$, where now the contour is deformed such that $\tilde y(x) \le y(x)$ is smaller. What this Lemma tells us is that, whenever $\Re f(z)$ reaches a maximum at some point $z_0 \in \gamma$, we can always deform the contour such that $\Re f(z)$ becomes strictly smaller near $z_0$, unless either $f$ is not holomorphic in $z_0$, or $\Im f'(z_0) = 0$. When all maxima along the contour lie inside the analytic region and have nonzero $\Im f'(z)$, we thus find that after repeating this procedure, we obtain a new contour where $\sup_{\tilde\gamma} f$ has strictly decreased. The signs of $\Im f'(z)$ tell us what is locally the direction of the deformation. \\
\\
We now move on to a practical Lemma that allows us to set a smooth lower bound on a possibly ill-behaved function.
\begin{lemma} \label{lm:smooth bound}
    Let $f(x):\mathbb R_{\ge 0} \to \mathbb R_{\ge 0}$ be a function with the following properties:
    \begin{itemize}
        \item $f(x)$ is monotonically increasing in $x$,
        \item $f(0)=0$,
        \item $f(x) > 0 \quad \forall x>0$.
    \end{itemize}
    There exists a smooth function $g(x):\mathbb R_{\ge 0} \to \mathbb R_{\ge 0}$ with the same three properties as $f(x)$, $g(x) \le f(x)$ for all $x\ge 0$, and $g^{(n)}(0)=0$ for all $n\ge 1$.
\end{lemma}
\begin{proof}
    Let $(x_n)_{n\ge 1}$ be a strictly decreasing sequence in $\mathbb R_{>0}$, with $x_n \to 0$. On each interval $[x_{n+1},x_n]$, we construct $\tilde g(x)$ as a smooth, monotonically increasing transition function with $\tilde g(x_n)=f(x_{n+1})$ and $\tilde g^{(m)}(x_n) = 0$ for all $m,n \ge 1$. We then have for all $x\in [x_{n+1},x_n]$:
    \begin{equation}
        \tilde g(x) \le \tilde g(x_n) = f(x_{n+1}) \le f(x) \,.
    \end{equation}
    We thus find a smooth, monotonically increasing function $\tilde g(x)$ on $\bigcup_{n \ge 1} [x_{n+1},x_n] = (0,x_1]$, which satisfies $0 <\tilde g(x) \le f(x)$ on this interval. Beyond $x_1$, we can simply set $\tilde g(x\ge x_1) = \tilde g(x_1)$.

    Finally, we multiply $g(x) := \tilde g(x)e^{-1/x}$. This will ensure the correct limits for $x\to 0^+$, where we simply define $g(0)=0$ and $g^{(m)}(0)=0$ for all $m\ge0$.

    To make sure that indeed $\lim_{x\to 0^+} g^{(m)}(x) = 0$, it is sufficient that $\tilde g^{(m)}(x)$ diverges at most polynomially for small $x$. This is for example the case if we choose $x_n = 1/n$, as we show now. Suppose $\phi(x)$ is a transition function that connects $\phi(0)=0$ and $\phi(1)=1$. Let
    \begin{equation}
        D_m := \sup_{x\in\mathbb [0,1]} |\phi^{(m)}(x)| < \infty \,,
    \end{equation}
    since $\phi(x)$ is smooth. When rescaled and resized to connect $\tilde g(x_{n+1}) = \tilde g_n$ to $\tilde g(x_n) = \tilde g_n + \Delta_n$, we find
    \begin{equation}
        \sup_{x\in\mathbb [x_{n+1},x_n]} |\tilde g^{(m)}(x)| = D_m \Delta_n (x_n-x_{n+1})^{-m} \,.
    \end{equation}
    Now, we choose $x_0 >0$ and look for $n_0 = \sup\{n\in\mathbb N:x_n>x_0\}$, such that $x_{n_0+1} \le x_0$. Then, we have
    \begin{multline}
        \sup_{x\ge x_0} |\tilde g^{(m)}(x)| \le \sup_{n\le n_0}\left\{\sup_{x\in[x_{n+1},x_n]} |\tilde g^{(m)}(x)|\right\} \\
        = D_m\sup_{n\le n_0}\left\{\Delta_n (x_n-x_{n+1})^{-m}\right\} \,.
    \end{multline}
    Firstly, because $\tilde g$ is monotonically increasing and bounded as $0 \le \tilde g(x) \le \tilde g(x_1)$, we also have $0 \le \Delta_n \le \tilde g(x_1)$. If we choose $x_n = 1/n$, we can furthermore compute $(x_n-x_{n+1})^{-1} = n(n+1) = x_n^{-1}(x_n^{-1}+1)$, and since $x_n > x_0 > 0$, we get
    \begin{multline}
        \sup_{x\ge x_0} |\tilde g^{(m)}(x)| \le D_m \left[\frac{1}{x_0}\left(\frac{1}{x_0}+1\right)\right]^m \sup_{n\le n_0} \Delta_n \\
        \le 2^m \tilde g(1) D_m /(x_0)^{2m} \equiv C_m/(x_0)^{2m} \,,
    \end{multline}
    where we used that $\tilde g^{(m)}(x) = 0$ for $x\ge x_1 = 1$.
\end{proof}

It is worth noting that the construction in this proof is not at all optimal. One would however expect it to be possible to find a smooth lower bound on $f(x)$ that lies arbitrarily close to the original function, outside of a small neighborhood around the discontinuities. This mathematical construction is therefore only required to make rigorous statements later on, and in practice we can just use the non-smooth function.

We now move on to proving the actual theorem of this section. We restate:

\Decay*
\begin{proof}
    First, we define $S(k;v) := ikv-i\omega(k)$, which is a one dimensional complex function with $v$ acting as a parameter. Since $S'(0;c_s)=0$, the point $S(0;c_s) = 0$ is a saddle point, so we cannot make any deformations to the contour $\gamma = (-\infty,\infty)$ that would lower the maximum of $\Re S(k;c_s)$ near $k=0$. We therefore set $\tilde\Gamma(c_s)=0$. 
    
    For values $v>c_s$, $\Im S'(0;c_s)>0$ so the deformation Lemma \ref{lm:deformation} tells us there does exist a deformation $\tilde\gamma_v$, parametrized by $\tilde y_v (k_R)$, which lowers $\Re S(k;v)$ in a neighborhood around $k_R\equiv\Re k=0$. Since $\Re S(k;v) < 0$ everywhere except for this neighborhood, we thus found a new contour where $\sup_{k\in\tilde\gamma_v} \Re S(k;v)  = -\tilde\Gamma(v) < 0$. The lemma also tells us that this contour lies entirely in the upper complex plane $\Im k\ge 0$. Since $\Re S(k;v) = -v \Im k +\Im \omega(k)$, increasing $v$ any further can therefore only decrease $\sup_{k\in\tilde\gamma_v} \Re S(k;v)$, not increase. As a result, it is possible to construct a function $\tilde\Gamma(v)=-\sup_{k\in\tilde\gamma_v} \Re S(k;v)>0$ for all $v>c_s$, that will always increase or remain constant with increasing values of $v$. 
    
    Since $\tilde\Gamma$ is not necessarily smooth, we instead use a smooth lower bound $\Gamma(v):\mathbb R_{\ge c_s} \to \mathbb R_{\ge 0}$, which has the required properties of the theorem for $v>c_s$, and furthermore has $\Gamma(c_s)=0$, $\Gamma^{(m)}(c_s)=0$ for all $m \in \mathbb N$. The existence of this function is guaranteed by Lemma \ref{lm:smooth bound}. In an analogous manner, we also find such a smooth lower bound for the decay rate $\Gamma(v)$ when $v< c_s$. Since the left and right limit in $v=c_s$ for the function and all its derivatives are $0$, the function connects both halves of the real line smoothly.

    Since the contour deformation $\tilde\gamma_v$ remains entirely in the holomorphic region of $S(k;v)$, given by $|k|< R$, the integral defining $G(t,v)$ is equivalent along each of these contours. With the parametrization $k = k_R + iy(k_R)$, we can write
    \begin{multline}
        G(t,v) = \frac{1}{2\pi}\theta(t)\int_{-\infty}^\infty e^{ik_R vt-\tilde y_v(k_R)vt-i\omega(k)t} \\
        \times g(k)(1+i\tilde y_v'(k_R))dk_R \,.
    \end{multline}
    In this equation, $\tilde y_v'(k_R)$ should be interpreted in a distributional sense. On a technical note: absolute continuity of $\tilde y_v$ is required to ensure $\tilde y_v'(k_R)$ is Lebesgue integrable and that its Lebesgue integral satisfies the fundamental theorem of calculus, but $\tilde y_v$ is easily lower/upper bounded by a smooth deformation, as also follows from Lemma \ref{lm:smooth bound}. We can now split the integral for $F(t,vt) = G(t,v)e^{\Gamma(v) t}$ into two parts: one for $|k_R|< R$ and one for $|k_R|>R$, so $F(t,x) = F_1(t,x) + F_2(t,x)$. The first term,
    \begin{multline}
        F_1(t,x\equiv vt) = \frac{\theta(t)}{2\pi} \int_{-R}^R e^{(\Gamma(v)-\tilde y_v(k_R)v-i\omega(k))t} \\
        \times g(k)(1+i\tilde y_v'(k_R))e^{ik_R x}dk_R \,,
    \end{multline}
    is the Fourier transform of a compact distribution, and so an entire (tempered) function in $x$. This is also true for all derivatives to $t$, and we conclude $F_1(t,x)$ is an analytic function in $(t,x)$. The argument of the exponential factor has real part $\Gamma(v) + \Re S(k;v) \le 0$ along the contour, which ensures $F_1(t,x)$ is tempered in $\mathbb R^2$.

    For the second term, we have
    \begin{multline}
        F_2(t,x) = \frac{\theta(t)}{2\pi} \int_{-\infty}^\infty e^{(\Gamma(v)-i\omega(k_R))t}\theta(k_R^2-R^2) \\
        \times g(k_R)e^{ik_R x}dk_R \,.
    \end{multline}
    Because $|g(k_R)|$ is bounded by a polynomial, and the exponential factor is $\le 1$ for all $k_R$ and $t\ge0$, this is again the Fourier transform of a tempered distribution. Applying it to a Schwartz function $\varphi(t,x)\in\mathcal S(\mathbb R^2)$, we get:
    \begin{equation}
        |\langle F_2, \varphi \rangle| \le \frac{1}{2\pi} \int_0^\infty \int_{-\infty}^\infty |g(k)|\cdot|\hat\varphi(t,k)| dt dk < \infty \,.
    \end{equation}
    We conclude that, as the sum of two tempered distributions, $F(t,x)$ is also tempered in $\mathbb R^2$.
\end{proof}
It is worth taking a moment to address the, perhaps confusing, notation when constantly switching between $v$ and $x$. For all $t>0$, they are interchangeable, with $x = vt$ and $v = x/t$. It is only when looking at the exact slice $t = 0$ of $\mathbb R^2$, that the use of $v$ becomes problematic. At this specific slice, we need to return to $(t,x)$ coordinates and obtain $G(0,x) = F(0,x)$. $G(t,x)$ and $F(t,x)$ are thus tempered distributions in $\mathbb R^2$, but $G(t,v)$ is not. The use of $v$ and $\Gamma(v)$ is entirely restricted to $t>0$, and is only used to set the exponential decay.\\

\end{document}